\documentclass[11pt]{article}

\input{Style.sty}

\title{
Space-time tradeoff   for sparse quantum state preparation
}
\author[1]{Jingquan Luo\thanks{luojq25@mail2.sysu.edu.cn}}
\author[1]{Guanzhong Li\thanks{ligzh9@mail2.sysu.edu.cn}}
\author[1, 2]{Lvzhou Li\thanks{Corresponding author: lilvzh@mail.sysu.edu.cn}}
\affil[1]{Institute of Quantum Computing and and Software, School of Computer Science and Engineering, Sun Yat-sen University, Guangzhou 510006, China}
\affil[2]{Quantum Science Center of Guangdong-Hong Kong-Macao Greater Bay Area, Shenzhen 518045, China}
\date{\today }

\begin{document}
\maketitle

\begin{abstract}

In this work, we investigate   the trade-off between the circuit depth and the number of ancillary qubits for preparing  sparse quantum states. We prove that any $n$-qubit $d$-spare quantum state (i.e., it has only $d$ non-zero amplitudes) can be prepared by a quantum   circuit with depth 
$O\left(\frac{nd \log m}{m \log m/n} + \log nd\right)$ using  $m\geq 6n$   ancillary qubits, which achieves the current best  trade-off between depth and ancilla number. In particular, when $m = \ta{\frac{nd}{\log d}}$, our result recovers the optimal  circuit depth $\Theta(\log nd)$  given in \hyperlink{cite.zhang2022quantum}{[Phys. Rev. Lett., 129, 230504(2022)]}, but using significantly fewer gates and ancillary qubits. 

\end{abstract}

\section{Introduction}

Quantum state preparation (QSP) is a fundamental task in quantum computing, enabling algorithms to initialize input states with prescribed amplitude distributions and playing a crucial role in various applications such as Hamiltonian simulation \cite{childs2018toward, low2017optimal, low2019hamiltonian, berry2015simulating} and quantum machine learning \cite{lloyd2014quantum, kerenidis2016quantum, kerenidis2019q, kerenidis2021quantum, rebentrost2014quantum}. Formally, given a normalized complex vector $\{ \alpha_i \}_{0 \le i < 2^n}$ representing an $n$-qubit quantum state, and an error bound $\varepsilon \in [0, 1)$, the goal is to construct a resource-efficient quantum circuit $U$ acting on $n + m$ qubits  such that
\begin{align}
U \ket{0}^{\otimes n} \ket{0}^{\otimes m} = \ket{\psi} \quad \text{and} \quad 
\left\| \ket{\psi} - \sum_{i=0}^{2^n - 1} \alpha_i \ket{i} \ket{0}^{\otimes m} \right\| \le \varepsilon, \label{eq:approx-qsp}
\end{align}
where $\ket{0}^{\otimes m}$ are ancillary qubits used to assist the preparation. When $\varepsilon = 0$, this task reduces to \emph{exact quantum state preparation}, where the output state must match the target amplitudes precisely.  Various standards can be employed to gauge the efficiency of a quantum circuit, including size (the gate-count), depth (the number of circuit layers  which corresponds to the  time for executing
the quantum circuit), space (the number of ancillary qubits), and more. Considering the  higher implementation cost of  some specific gates such  the CNOT-gate and T-gate, the count of them is also utilized as a measure of efficiency.


Approximate quantum state preparation ($\varepsilon > 0$) has been extensively studied, mainly driven by two considerations. 
First, in fault-tolerant quantum computation, one typically works with a finite universal gate set such as Clifford+$T$, which renders the exact preparation of arbitrary amplitudes infeasible.
A sequence of works has focused on this setting~\cite{shende2005synthesis, sun2023asymptotically, zhang2022quantum, zhang2024circuit, zhang2024practical, gui2024spacetime, low2024trading, gosset2024quantum}. In particular, the best known circuit size is $\bo{2^n \log \frac{1}{\varepsilon}}$~\cite{zhang2022quantum, zhang2024circuit, zhang2024practical} and the circuit depth has been reduced to the asymptotically optimal bound of $\ta{n + \log \frac{1}{\varepsilon}}$~\cite{zhang2024practical, gui2024spacetime}. Moreover, the $T$-count has also been shown to be asymptotically optimal at $\ta{\sqrt{2^n \log \frac{1}{\varepsilon}} + \log \frac{1}{\varepsilon}}$~\cite{gosset2024quantum}.
Second, in many practical scenarios, the target state possesses intrinsic structure---for example, the amplitudes may be defined analytically, such as by a probability density function. In such cases, one can exploit this structure to design efficient approximate state preparation algorithms using matrix product states (MPS)~\cite{holmes2020efficient, melnikov2023quantum, iaconis2024quantum, gonzalez2024efficient, manabe2025state, gundlapalli2022deterministic}, quantum singular value transformation (QSVT)~\cite{gonzalez2024efficient, mcardle2022quantum}, and series expansion techniques~\cite{rosenkranz2025quantum, zylberman2024efficient, moosa2023linear}. Another prominent line of research leverages variational circuits~\cite{marin2023quantum, nakaji2022approximate, zoufal2019quantum}, which optimize parameterized quantum circuits to approximate the target state. 
Although these methods could be effective in practice, they generally lack rigorous complexity guarantees and often require problem-specific analysis. In addition, several works have considered the oracle-access model, where the amplitude information is not given analytically but instead retrieved via a unitary oracle~\cite{sanders2019black, bausch2022fast, wang2021fast, rattew2022preparing}.

In this work, we focus on the problem of exact quantum state preparation, corresponding to the case $\varepsilon = 0$ in Eq.~\eqref{eq:approx-qsp}. We consider a standard gate set consisting of arbitrary single-qubit gates and CNOT gates, which is commonly used in circuit synthesis~\cite{shende2005synthesis, bergholm2005quantum}. 
A sequence of works has established that the optimal circuit size is $\Theta(2^n)$~\cite{grover2002creating,shende2005synthesis,bergholm2005quantum,plesch2011quantum,iten2016quantum}. Several studies have focused on reducing circuit depth by leveraging ancillary qubits~\cite{sun2023asymptotically, yuan2023optimal, zhang2022quantum, rosenthal2021query}, and have established the asymptotically optimal depth-ancilla trade-off of $\Theta\left(n + \frac{2^n}{n + m}\right)$ when $m$ ancillary qubits are available~\cite{sun2023asymptotically, yuan2023optimal}. These results characterize the fundamental limits of exact state preparation from a worst-case perspective.

When the target state contains at most $d \ll 2^n$ nonzero amplitudes, the task is referred to as \emph{sparse quantum state preparation} (SQSP).
Formally, given a set of amplitude-basis pairs $\{(\alpha_i, q_i)\}_{i=0}^{d-1}$ with $\alpha_i \in \mathbb{C}$, $q_i \in \{0,1\}^n$, and $\sum_i |\alpha_i|^2 = 1$, the goal is to construct a quantum circuit $U$ such that
\begin{align}
U \ket{0}^{\otimes n} \ket{0}^{\otimes m} = \sum_{i=0}^{d-1} \alpha_i \ket{q_i} \ket{0}^{\otimes m}. \label{eq:sqsp}
\end{align}
This setting arises naturally in quantum algorithms where the input data is sparse. For example, SQSP plays a crucial role in constructing the block-encoding of a sparse matrix \cite{clader2022quantum, zhang2024circuit}, which is typically treated as a black box in quantum linear systems solvers~\cite{harrow2009quantum, costa2022optimal} and quantum singular-value transformation~\cite{martyn2021grand, gilyen2019quantum}. Several recent works have focused on optimizing the circuit size for SQSP~\cite{gleinig2021efficient, malvetti2021quantum, ramacciotti2023simple, mozafari2022efficient, de2020circuit, de2022double, mao2024towards, luo2024circuit}. In particular, Ref. \cite{luo2024circuit} established a nearly optimal  circuit size of $O({\frac{nd}{\log (n + m)} + n})$ when $m$ ancillary qubits are available.
We note that some recent works have explored the use of intermediate measurements to reduce the depth of (sparse) quantum state preparation circuits~\cite{yeo2025reducing, buhrman2024state, zi2025constant}. However, such approaches fall outside the scope of this paper.

In this paper, we investigate the circuit depth  of exact sparse quantum state preparation, with particular attention to the trade-off between the depth and the number of ancillary qubits. Under current experimental constraints, quantum coherence time is limited, making circuit depth a critical bottleneck. Meanwhile, the number of available qubits continues to grow, motivating the exploration of space-time trade-offs in quantum circuit design.
To address this challenge, we develop new circuit construction for SQSP that achieves lower circuit depth by utilizing additional ancillary qubits. Our results advance the theoretical understanding of SQSP and contribute toward making it more practical for near-term quantum devices.

\subsection{Contributions}
In this paper, we aim to reduce the circuit depth for SQSP with the help of ancillary qubits, achieving a trade-off between the circuit depth and the number of ancillary qubits, as follows:
\begin{theorem}\label{thm:tradeoff}
Any $n$-qubit $d$-sparse quantum state can be prepared by a circuit of depth $\bo{\frac{nd\log m}{m\log m/n}+\log nd}$ and size $\bo{\frac{nd}{\log m/n} + \frac{nd}{\log d}}$, by using  $m \geq 6n$ ancillary qubits.
\end{theorem}

There are few related works  on circuit depth for SQSP \cite{sun2023asymptotically, zhang2024circuit, zhang2022quantum}.  Ref.~\cite{sun2023asymptotically} primarily investigated the circuit depth for general quantum state preparation and mentioned that their method can be extended to the sparse case. In the full version~\cite{sun2021asymptotically}, the authors provided a relatively loose upper bound of $O(n \log dn + \frac{dn^2 \log d}{n + m})$. Although the focus of  Ref. \cite{zhang2024circuit} is about approximately preparing sparse quantum states using Clifford+$T$ gates,  
it is straightforward to extend the results to the exact  setting: when $m \in \om{n}$, 
 there exists a circuit for SQSP whose size is $\bo{nd\log d}$ and depth is $\bo{\frac{nd\log d \log m}{m} + \log nd}$.
In comparison, \cref{thm:tradeoff} reduces both the size and the depth by a factor of $\log d$ when $m$ is of moderate size.
Note that $\log d$ can be as large as $\ta{n}$.

Ref.~\cite{zhang2022quantum}  is one of the few papers focusing on the circuit depth of SQSP in the exact setting. When $m = \ta{\frac{nd}{\log d}}$, Theorem~\ref{thm:tradeoff} achieves the optimal circuit depth $\Theta(\log nd)$ with circuit size $\bo{\frac{nd}{\log d}}$, matching the circuit depth achieved in Ref.~\cite{zhang2022quantum}, which, however, requires $O(nd \log d)$ elementary gates and $O(nd \log d)$ ancillary qubits. Therefore, compared to the result in Ref.~\cite{zhang2022quantum}, our result achieves the same asymptotically optimal circuit depth while using significantly fewer ancillary qubits and  gates.

\subsection{Discussion}

We achieve an ancilla-depth trade-off for sparse quantum state preparation in this paper;
however, several interesting problems remain open and are worthy of further discussion:
\begin{itemize}
    \item Can we further improve the ancilla-depth trade-off, or alternatively, prove a matching lower bound? It has been shown that the circuit size lower bound for sparse quantum state preparation is $\Omega\left(\frac{nd}{\log(m + n) + \log d} + n\right)$~\cite{luo2024circuit}, and a general circuit depth lower bound of $\Omega(\log nd)$ has been established in~\cite{zhang2022quantum}. Together, these results imply that any SQSP circuit using $m$ ancillary qubits must have circuit depth at least $\Omega\left(\frac{nd}{(m + n)(\log(m + n) + \log d)} + \log nd\right)$ in the worse case. This raises the natural question of whether the current ancilla-depth trade-offs are tight or can still be improved.
 
    \item We say a Boolean function $f:\cbra{0, 1}^n \xrightarrow{} \cbra{0, 1}^{\tilde{n}}$ is a sparse Boolean function if it has at most $d$ inputs that have non-zero outputs. 
    It was proved in \cite{zhang2024circuit} that a sparse Boolean function could be implemented by a circuit of size $\bo{n\tilde{n}d}$ and depth $\bo{\frac{n\tilde{n}d\log m}{m}}$ given $\om{n} \leq m \leq \bo{n\tilde{n}d}$ ancillary qubits. An interesting problem is to improve the circuit size and the ancilla-depth trade-off for implementing sparse Boolean functions.
\end{itemize}

\subsection{Organization}
In Section~\ref{section:pre}, we recall some notations and introduce some helpful lemmas.
In Section~\ref{sec:tradeoff}, we prove the main theorem of the paper, i.e., Theorem~\ref{thm:tradeoff}.
Some conclusions are made in Section \ref{sec:conclusion}.

\section{Preliminaries}\label{section:pre}
In this paper, all logarithms are assumed to be base 2. We use subscripts to denote the number of qubits in a quantum register. For example, $\ket{\cdot}_n$ represents a register consisting of $n$ qubits. Given an $n$-bit integer $x = x_0 \dots x_{n-1}$, for any $0 \leq j < k \leq n$, we denote the substring $x_j \dots x_{k-1}$ as $x[j: k]$. The unary encoding $e_x$ of $x$ is a $2^n$-bit string, where the $x$-th bit is set to 1, and all the other bits are set to 0.

In the following, we introduce some helpful lemmas. The first lemma is on the depth of a $(n+1)$-qubit parity gate.
\begin{lemma}[\cite{moore2001parallel}]\label{lemma:parallelcnot}
    Let $C$ be a quantum circuit consisting of $n$ CNOT gates with identical target qubits but different control qubits. There exists an equivalent circuit of depth $\bo{\log n}$ and size $\bo{n}$, without ancillary qubits.
\end{lemma}

To facilitate the parallel execution of quantum gates, it is necessary to efficiently replicate the contents of a register, as shown in the following lemma:

\begin{lemma}[{\cite[Fanout Gates]{moore2001parallel, sun2023asymptotically}}]\label{lemma:copy}
There is a circuit of depth $\bo{\log n}$ and consisting of $\bo{n}$ CNOT gates, which implements the following transformation:
\begin{align}
    \ket{x} \ket{b_1}\ket{b_2}\dots \ket{b_n} \xrightarrow{} \ket{x} \ket{b_1\oplus x}\ket{b_2\oplus x}\dots \ket{b_n\oplus x}.
\end{align}
where $x, b_1, \dots, b_n \in \cbra{0, 1}$. And if $b_i = 0$ for all $1 \leq i \leq n$, then the circuit consists of $n$ CNOT gates.
\end{lemma}

Apart from CNOT circuits, multi-control Toffoli gates, which compute the function $\text{AND}_n$ of the $n$ control qubits and add the result to the target qubit, are another common gates in quantum circuit synthesis. Recently, \cite{nie2024quantum} proved that an $n$-qubit Toffoli gate can be realized by a circuit of logarithmic depth.

\begin{lemma}[{\cite[Theorem 1, Theorem 6]{nie2024quantum}}]\label{lemma:toffoli}
    An $n$-qubit Toffoli gate can be implemented by a circuit of depth $\bo{\log n}$ and size $\bo{n}$ with one ancillary qubit. If no ancillary qubits are available, the depth is $\bo{\log^2 n}$ and the circuit size is $\bo{n}$.
\end{lemma}
 The last lemma shows an efficient transformation between the binary encoding and the unary encoding.
\begin{lemma}[{\cite[Lemma 22]{sun2023asymptotically}}]\label{lemma:unary2binary2}
    There exists a quantum circuit that achieves the following transformation:
    \begin{align}
        \ket{e_i}_{2^n}\ket{0}_n \xrightarrow{} \ket{0^{2^n}}_{2^n} \ket{i}_{n} \text{, for all $0 \leq i \leq 2^n-1$}.
    \end{align}
    The circuit is of size $\bo{2^n}$ and depth $\bo{n}$, using $2^{n+1}$ ancillary qubits.
\end{lemma}

\section{Sparse Quantum State Preparation}\label{sec:tradeoff}
 In this section, we prove \cref{thm:tradeoff}. In \cref{framework}, we give a brief introduction to the framework of the proposed algorithm.  In \cref{details}, we dive into the details and analyze the complexity of the algorithm. Finally, in \cref{sec:proof} we provide the deferred proofs of two key lemmas.
 
\subsection{Framework}\label{framework}

Our algorithm follows a two-step manner: a $\ceil{\log d}$-qubit quantum state $\sum_i \alpha_i \ket{i}$ is first prepared, and then it is transformed into the target quantum state. 
The first step is a task of general quantum state preparation, which has been extensively investigated and could be implemented by a circuit of size $\bo{d}$ and depth $\bo{\frac{d}{n+m}+\log d}$ with $m$ ancillary qubits~\cite{sun2023asymptotically, yuan2023optimal}. 
The first step is not the main source of complexity, thus we focus on the second step from now on.
However, rather than transforming directly into the target state, we introduce an intermediate state, and then we follow the two-phase workflow:  
\begin{align}\label{equ:workflow}
   \sum_{i=0}^{d-1} \alpha_i \ket{i}_{\ceil{\log d}} \xrightarrow[]{\texttt{{Phase 1}} } \sum_{i=0}^{d-1} \alpha_i \rbra*{\bigotimes_{j=0}^{\frac{n}{r}-1} \ket{e_{q_i[jr: (j+1)r]}}_{2^r}} \xrightarrow[]{\texttt{{Phase 2}} } \sum_{i=0}^{d-1} \alpha_i \ket{q_i}_{n},
\end{align}
where $r$ is a parameter to be determined later, and without loss of generality, $r$ is  required to divide $n$. $\bigotimes_{j=0}^{\frac{n}{r}-1} \ket{e_{q_i[jr: (j+1)r]}}_{2^r}$ is the $(n,r)$-unary encoding of  $q_i$, as defined in \cref{def:nr}.  This intermediate state was previously used in \cite{luo2024circuit} to reduce the circuit size for SQSP. 
In summary, in \texttt{{Phase 1}}, we transform the initial state to the intermediate state, which is further transformed to the target state in \texttt{{Phase 2}}. In the process, we optimize the circuit depth by leveraging auxiliary qubits to maximize the parallel execution of quantum gates, obtaining the trade-off between depth and ancilla number.



\begin{definition}\label{def:nr}
Given an $n$-bit integer $x = x_0x_1\cdots x_{n-1}$ and another integer $r > 0$ such that $r$ divides $n$, the $(n, r)$-unary encoding of  $x$ is a $\frac{n2^r}{r}$-bit string $e_{x[0:r]} e_{x[r:2r]} \cdots e_{x[n-r:n]}$. In particular, when $r = n$, the $(n, n)$-unary encoding corresponds to the standard unary encoding.
\end{definition}

 We give a concrete example to illustrate the notations and our workflow.
\begin{example}\label{examp:1}
    Let $d = 4$, $n = 8$, $r = 2$, with $q_0 = 11011000$, $q_1 = 00011001$, $q_2 = 10001111$, and $q_3 = 01011011$. Firstly, we prepare a two-qubit state as follows:
    \begin{align}
        \alpha_0 \ket{0}_2 + \alpha_1 \ket{1}_2 + \alpha_2 \ket{2}_2 + \alpha_3 \ket{3}_2,
    \end{align}
    which is further transformed to the $16$-qubit intermediate state in \texttt{{Phase 1}}:
    \begin{align}
    & \alpha_0  \ket{0001\ 0100\ 0010\ 1000}
            + \alpha_1 \ket{1000\ 0100\ 0010\ 0100} \\
    +\ & \alpha_2 \ket{0010\ 1000\ 0001\ 0001} + \alpha_3 \ket{0100\ 0100\ 0010\ 0001}.
    \end{align}
    After \texttt{{Phase 2}}, we obtain the target state:
    \begin{align}
    \alpha_0  \ket{11011000}
            + \alpha_1 \ket{00011001} 
    + \alpha_2 \ket{10001111} + \alpha_3 \ket{01011011}.
    \end{align}
\end{example}

\subsection{Algorithm}\label{details}
To establish the correctness of our algorithm, we demonstrate how to implement \texttt{{Phase 1}} and \texttt{{Phase 2}} as defined in \cref{equ:workflow}. We begin with \texttt{{Phase 1}}, for which we prove the following lemma:

\begin{lemma}\label{lemma:phase1}
   Given an integer $r > 0$ such that $r$ divides $n$,  for any $0 < k \leq d$ being a power of two, \texttt{{Phase 1}}, as defined in \cref{equ:workflow}, can be implemented by a quantum circuit of depth $\bo{\frac{d\log nk}{k}}$ and size $\bo{\frac{nd}{r} + \frac{d\log d}{k}}$. The total number of qubits  is $\ceil{\log d} + 4k + \frac{nk}{r} + \frac{n2^r}{r}$.
\end{lemma}

\begin{proof}
Let ${\ell_n} \coloneqq \ceil{\log d}$ and $\ell_k \coloneqq \log k$. 
The  set $\mathcal{S} \coloneqq \cbra{0, \dots, d-1}$ is partitioned into $\ceil{\frac{d}{k}}$ subsets $\cbra{\mathcal{S}_j}_{0 \leq j < \ceil{\frac{d}{k}}}$, where all the  elements within the same subset $\mathcal{S}_j$ share the same prefix $j$ in the $({\ell_n} - \ell_k)$ most significant bits. Formally, for $0 \leq j < \ceil{\frac{d}{k}} - 1$, we define  
\begin{align}\label{equ:sj1}
    \mathcal{S}_j \coloneqq \cbra{jk, \dots, \rbra{j+1}k - 1},  
\end{align}  
and for the last subset, we define  
\begin{align}\label{equ:sj2}
    \mathcal{S}_{\ceil{\frac{d}{k}}-1} \coloneqq \cbra{(\ceil{\frac{d}{k}}-1)k, \dots, d-1}.  
\end{align}  

\begin{example}
    Let $d = 7$, $k=4$. Then we have ${\ell_n} = 3$, $\ell_k = 2$, $\mathcal{S} = \cbra{0, 1, 2, 3, 4, 5, 6}$, $\mathcal{S}_0 = \cbra{0, 1, 2, 3}$ and $\mathcal{S}_1 = \cbra{4, 5, 6}$, as shown in \cref{table:s}.

\end{example}
\begin{table}[htbp]
\centering
\caption{Illustration of the notations.}
\label{table:s}
\begin{threeparttable}
\begin{tabular}{c|c|c|c}
\toprule
                    & $i$ & $i[0:{\ell_n}-\ell_k]$ & $i[{\ell_n}-\ell_k:{\ell_n}]$   \\ \midrule
\multirow{4}{*}{$\mathcal{S}_0$} & 0     & 0 & 00   \\ \cline{2-4} 
                    & 1     & 0 & 01   \\ \cline{2-4} 
                    & 2     & 0 & 10   \\ \cline{2-4} 
                    & 3     & 0 & 11   \\ \midrule
\multirow{4}{*}{$\mathcal{S}_1$} & 4     & 1 & 00   \\ \cline{2-4} 
                    & 5     & 1 & 01   \\ \cline{2-4} 
                    & 6     & 1 & 10   \\  \bottomrule
\end{tabular}
\end{threeparttable}

\end{table}

The procedure for \texttt{{Phase 1}} is presented in \cref{alg:nr-encoding}. We start with the initial state
\begin{align}
    \sum_{0 \leq i < d} \alpha_i \ket{i}_{{\ell_n}} \ket{0}_k \ket{0}_k \ket{0}_{n2^r/r},
\end{align}
where we omit some ancillary qubits for brevity. 
At first, we make the transformation on the first ${\ell_n} + k$ qubits that transforms the $\ell_k$ least significant bits of $\ket{i}_{{\ell_n}}$ into its unary encoding in a new register:
\begin{equation}\label{equ:phase1-step1}
\begin{aligned}
    \sum_{i = 0}^{d-1} \alpha_i \ket{i}_{{\ell_n}} \ket{0}_{k} 
     & = \sum_{j=0 }^{ \ceil{\frac{d}{k}}-1} \sum_{i=0}^{ \abs{\mathcal{S}_j}-1} \alpha_{jk+i} \ket{j}_{{\ell_n}-\ell_k} \ket{i}_{\ell_k} \ket{0}_k\\
    & \xrightarrow{} \sum_{j=0 }^{ \ceil{\frac{d}{k}}-1} \sum_{i=0}^{ \abs{\mathcal{S}_j}-1} \alpha_{jk+i} \ket{j}_{{\ell_n}-\ell_k} \ket{0}_{\ell_k} \ket{e_i}_k,
\end{aligned}
\end{equation}
where in the first equation, these basis elements are  partitioned into $\ceil{\frac{d}{k}}$ subsets such that the  elements within the same subset $\mathcal{S}_j$ share the same prefix $j$ in the $({\ell_n} - \ell_k)$ most significant bits as defined in  \cref{equ:sj1} and \cref{equ:sj2}. 
According to \cref{lemma:unary2binary2}, this transformation can be implemented by a quantum circuit of size $\bo{k}$ and depth $\bo{\log k}$ using additional $2k$ ancillary qubits.

\begin{algorithm}[htbp]
\caption{ $(n, r)$-Unary Encoding}\label{alg:nr-encoding} 
\begin{algorithmic}[1]
    \REQUIRE {$ \mathcal{P} = \cbra*{\rbra*{\alpha_i, q_i}}_{0 \leq i \leq d - 1}$.}
    \ENSURE {Quantum state $\sum_{i=0}^{d-1} \alpha_i \rbra*{\bigotimes_{j=0}^{\frac{n}{r}-1} \ket{e_{q_i(jr, (j+1)r)}}}$.}
    \STATE Let the initial state  be $\sum_i \alpha_i \ket{i}_{{\ell_n}} \ket{0}_k \ket{0}_k \ket{0}_{n2^r/r}$ ;
    \STATE Let $\cbra{S_j}_{0 \leq j < \ceil{\frac{d}{k}}}$ defined as \cref{equ:sj1} and \cref{equ:sj2} ;
    \STATE Apply the transformation defined as in \cref{equ:phase1-step1} on the first ${\ell_n} + k$ qubits;
    \FOR{$j=0$ to $\ceil{\frac{d}{k}}-1$} 
        \STATE Apply circuit $\mathcal{C}_j^1$ on the first ${\ell_n} + 2k$ qubits according to \cref{lemma:encodeunary} ;
        \STATE Apply circuit $\mathcal{C}_j^2$ on the last $k + \frac{n2^r}{r}$ qubits, where $\mathcal{C}_j^2$ is defined according to \cref{lemma:u2nr} with $x_i \coloneqq q_{jk+i}$ for $0 \leq i \leq \abs{\mathcal{S}_j}-1$ ;
    \ENDFOR
\end{algorithmic}

\end{algorithm}

Then, for each $j$ from $0$ to $\ceil{\frac{d}{k}}-1$, we apply the circuits described in \cref{lemma:encodeunary} and  \cref{lemma:u2nr}. The proofs of the two lemmas are deferred to \cref{sec:proof}.

\begin{lemma}\label{lemma:encodeunary}
    In the $j$-th iteration, there exists a quantum circuit $\mathcal{C}_j^1$ achieving the following transformation for the basis states corresponding to $\mathcal{S}_j$ while keeping other basis states unchanged:
    \begin{align}\label{equ:tounary}
    \ket{j}_{{\ell_n}-\ell_k}\ket{0}_{\ell_k} \ket{e_i}_k \ket{0}_k 
    \xrightarrow[]{} 
    \ket{0}_{{\ell_n}}\ket{0}_k \ket{e_i}_k, \quad \text{for } 0 \leq i < k.
\end{align}
The circuit is of size $O({\ell_n} - \ell_k + k)$ and depth $O(\log ({\ell_n} - \ell_k) + \log k)$, using additional $k$ ancillary qubits. 
\end{lemma}

\begin{lemma}\label{lemma:u2nr}
   Given positive integers $n$, $r$, $k$, a set of $n$-bit integers $\mathcal{X} = \cbra{x_i}$ of size at most $k$,
    there exists a quantum circuit $\mathcal{C}^2$ implementing the following transformation:
    \begin{align}
        & \ket{e_i}_k \ket{0}_{n2^r/r}   \xrightarrow[]{}  \ket{0}_k \rbra*{\bigotimes_{j=0}^{\frac{n}{r}-1} \ket{e_{x_i[jr: (j+1)r]}}_{2^r}} \text{, for all $0 \leq i \leq \abs{\mathcal{X}}-1$}. \label{equ:cc21} \\ 
        & \ket{0}_k \rbra*{\bigotimes_{j=0}^{\frac{n}{r}-1} \ket{e_{x[jr: (j+1)r]}}_{2^r}}   \xrightarrow[]{}  \ket{0}_k \rbra*{\bigotimes_{j=0}^{\frac{n}{r}-1} \ket{e_{x[jr: (j+1)r]}}_{2^r}} \text{, if $x \notin \mathcal{X}$}, \label{equ:cc22} \\
        &  \ket{0}_k \ket{0}_{n2^r/r}  \xrightarrow[]{}  \ket{0}_k \ket{0}_{n2^r/r}. 
        \label{equ:cc23}
    \end{align}
    The circuit is of size $\bo{\frac{nk}{r}}$ and depth $\bo{\log nk}$, and using additional $(\frac{nk}{r}+k)$ ancillary qubits. 
\end{lemma}

In the $j$-th iteration, the $(n, r)$-unary representation of each $q_i$ for $i \in \mathcal{S}_j$ is encoded into a register of $\frac{n2^r}{r}$ qubits. More specifically, applying the circuit given in \cref{lemma:encodeunary} on the first ${\ell_n} + 2k$ qubits, the state corresponding to the set $\mathcal{S}_j$ undergoes the following transformation:
\begin{align} \sum_{i=0}^{ \abs{\mathcal{S}_j}-1} \alpha_{jk+i} \ket{j}_{{\ell_n}-\ell_k} \ket{0}_{\ell_k} \ket{e_i}_k \ket{0}_k \xrightarrow{} \sum_{i=0}^{ \abs{\mathcal{S}_j}-1} \alpha_{jk+i} \ket{0}_{{\ell_n}} \ket{0}_k \ket{e_i}_k.
 \end{align}
Furthermore,  applying circuit $ \mathcal{C}_j^2$ defined according to \cref{lemma:u2nr} with $x_i \coloneqq q_{jk+i}$ for $0 \leq i \leq \abs{\mathcal{S}_j}-1$, the above state joint with $\ket{0}_{n2^r/r}$  is transformed into the following state:
\begin{align} \sum_{i=0}^{ \abs{\mathcal{S}_j}-1} \alpha_{jk+i} \ket{0}_{{\ell_n}} \ket{0}_k \ket{0}_k \rbra*{\bigotimes_{j'=0}^{\frac{n}{r}-1} \ket{e_{q_{jk+i}[j'r: (j'+1)r]}}_{2^r}}.
 \end{align}
 \cref{equ:cc22} and \cref{equ:cc23} ensure that $\mathcal{C}_j^2$ would perserve other basis states not corresponding to the set $\mathcal{S}_j$.

 Therefore,  after  
 $j$ running from $0$ to $\ceil{\frac{d}{k}}-1$, by tracing out the ancillary qubits, we obtain the target state:
\begin{align}\sum_{i=0}^{d-1} \alpha_i \rbra*{\bigotimes_{j=0}^{\frac{n}{r}-1} \ket{e_{q_i[jr: (j+1)r]}}_{2^r}}. 
 \end{align}

Next, we analyze the overall circuit complexity. The circuit size is given by  
\begin{align}
O(k) + O({\ell_n} - \ell_k + k) \cdot \ceil*{\frac{d}{k}} + O\left(\frac{nk}{r}\right) \cdot \ceil*{\frac{d}{k}} = O\left(\frac{nd}{r} + \frac{d\log d}{k}\right).
\end{align}  
The circuit depth is  
\begin{align}
O(\log k) + O(\log ({\ell_n} - \ell_k) + \log k) \cdot \ceil*{\frac{d}{k}} + O(\log nk) \cdot \ceil*{\frac{d}{k}} = O\left(\frac{d\log nk}{k}\right).
\end{align}  
The total number of qubits required is  
\begin{align}
\ceil{\log d} + 3k + \frac{nk}{r} + k + \frac{n2^r}{r} = \ceil{\log d} + 4k + \frac{nk}{r} + \frac{n2^r}{r}.
\end{align}
\end{proof}

We give an example to illustrate the workflow of \cref{lemma:phase1}.
\begin{example}
    Let $d = 4$, $n = 8$, $r = 2$, with $q_0 = 11011000$, $q_1 = 00011001$, $q_2 = 10001111$, and $q_3 = 01011011$, as in \cref{examp:1}. Let $k = 2$. After the transformation defined in \cref{equ:phase1-step1}, the state of the first $\ell_d + k$ qubits is
    \begin{align}
        \ket{\psi_1} = \alpha_0 \ket{0}_1 \ket{0}_1\ket{10}_2 + \alpha_1 \ket{0}_1 \ket{0}_1\ket{01}_2 
        + \alpha_2 \ket{1}_1 \ket{0}_1\ket{10}_2 + \alpha_3 \ket{1}_1 \ket{0}_1\ket{01}_2.
    \end{align}
    According to the definition of $\mathcal{S}_j$, we have $\mathcal{S}_0 = \cbra{0, 1}$ and $\mathcal{S}_1 = \cbra{2, 3}$. In the first iteration, after the transformation according to \cref{lemma:encodeunary}, the state of the first $\ell_d + 2k$ qubits is
    \begin{equation}
    \begin{aligned}
        \ket{\psi_2} & = \alpha_0 \ket{00}_2 \ket{00}_2 \ket{10}_2 + \alpha_1 \ket{00}_2 \ket{00}_2\ket{01}_2 \\
        &+ \alpha_2 \ket{1}_1 \ket{0}_1\ket{10}_2\ket{00}_2 + \alpha_3 \ket{1}_1 \ket{0}_1\ket{01}_2\ket{00}_2.
    \end{aligned}
    \end{equation}
    Then apply a circuit according to \cref{lemma:u2nr} for $\mathcal{S}_0$, and the state is
    \begin{equation}
    \begin{aligned}
        \ket{\psi_3} & = \alpha_0 \ket{00}_2 \ket{00}_2 \ket{00}_2 \ket{0001\ 0100\ 0010\ 1000}_{16} \\
        &+ \alpha_1 \ket{00}_2 \ket{00}_2\ket{00}_2 \ket{1000\ 0100\ 0010\ 0100}_{16}\\
        &+ \alpha_2 \ket{1}_1 \ket{0}_1\ket{10}_2\ket{00}_2 \ket{0000\ 0000\ 0000\ 0000}_{16} \\&+ \alpha_3 \ket{1}_1 \ket{0}_1\ket{01}_2\ket{00}_2 \ket{0000\ 0000\ 0000\ 0000}_{16}.
    \end{aligned}
    \end{equation}
    In the second iteration, after the transformation defined in \cref{equ:tounary}, we have
    \begin{equation}
    \begin{aligned}
        \ket{\psi_4} & = \alpha_0 \ket{00}_2 \ket{00}_2 \ket{00}_2 \ket{0001\ 0100\ 0010\ 1000}_{16} \\
        &+ \alpha_1 \ket{00}_2 \ket{00}_2\ket{00}_2 \ket{1000\ 0100\ 0010\ 0100}_{16}\\
        &+ \alpha_2 \ket{00}_2 \ket{00}_2 \ket{10}_2 \ket{0000\ 0000\ 0000\ 0000}_{16} \\&+ \alpha_3 \ket{00}_2 \ket{00}_2 \ket{01}_2 \ket{0000\ 0000\ 0000\ 0000}_{16}.
    \end{aligned}
    \end{equation}
    Then apply a circuit according to \cref{lemma:u2nr} for $\mathcal{S}_1$, and we have 
    \begin{equation}
    \begin{aligned}
        \ket{\psi_5} & = \alpha_0 \ket{00}_2 \ket{00}_2 \ket{00}_2 \ket{0001\ 0100\ 0010\ 1000}_{16} \\
        &+ \alpha_1 \ket{00}_2 \ket{00}_2\ket{00}_2 \ket{1000\ 0100\ 0010\ 0100}_{16}\\
        &+ \alpha_2 \ket{00}_2 \ket{00}_2 \ket{00}_2 \ket{0010\ 1000\ 0001\ 0001}_{16} \\&+ \alpha_3 \ket{00}_2 \ket{00}_2 \ket{00}_2 \ket{0100\ 0100\ 0010\ 0001}_{16}.
    \end{aligned}
    \end{equation}
After the iteration and tracing out the first $6$ qubits, we obtain the target state.
\end{example}

Now, we are ready to prove the main theorem.  

\begin{theorem}[Restatement of \cref{thm:tradeoff}]
Any $n$-qubit $d$-sparse quantum state can be prepared by a circuit of depth $\bo{\frac{nd\log m}{m\log m/n}+\log nd}$ and size $\bo{\frac{nd}{\log m/n} + \frac{nd}{\log d}}$, by using  $m \geq 6n$ ancillary qubits.
\end{theorem}
\begin{proof}

As mentioned before, our algorithm follows the workflow:
\begin{align}
   \sum_{i=0}^{d-1} \alpha_i \ket{i}_{\ceil{\log d}} \xrightarrow[]{\texttt{{Phase 1}} } \sum_{i=0}^{d-1} \alpha_i \rbra*{\bigotimes_{j=0}^{\frac{n}{r}-1} \ket{e_{q_i[jr: (j+1)r]}}_{2^r}} \xrightarrow[]{\texttt{{Phase 2}} } \sum_{i=0}^{d-1} \alpha_i \ket{q_i}_{n}.
\end{align}
So far, we have shown how to implement \texttt{{Phase 1}} and it remains to implement \texttt{{Phase 2}}. Note that $\ket{e_{q_i[jr: (j+1)r]}}$ is the unary encoding of the $r$-bit integer ${q_i[jr: (j+1)r]}$. Thus, \texttt{{Phase 2}} is to apply \cref{lemma:unary2binary2} on all $\cbra{\ket{e_{q_i[jr: (j+1)r]}}}_{0 \leq j \leq \frac{n}{r}-1}$ in parallel.

According to \cref{lemma:phase1}, \texttt{{Phase 1}} could be implemented by a quantum of depth $\bo{\frac{d\log nk}{k}}$ and size $\bo{\frac{nd}{r} + \frac{d\log d}{k}}$. All the qubits involved in \texttt{{Phase 1}} should be considered as the ancillary qubits, which counts to $\ceil{\log d} + 4k + \frac{nk}{r} + \frac{n2^r}{r}$. According to \cref{lemma:unary2binary2}, \texttt{{Phase 2}} can be realized by a circuit of depth $\bo{r}$ and size $\bo{\frac{n2^r}{r}}$, utilizing $\frac{n2^{r+1}}{r}$ ancillary qubits. Including the qubits encoding the $(n, r)$-unary representation, the number of the ancillary qubits involved in \texttt{{Phase 2}} is $\frac{3n2^{r}}{r}$. Therefore, the parameters $k$ and $r$ should satisfy the following inequalities:
    \begin{align}
        m &\geq \ceil{\log d} + 4k + \frac{nk}{r} + \frac{n2^r}{r}, \label{m-ineq1}\\
        m &\geq \frac{3n2^{r}}{r}.\label{m-ineq}
    \end{align}
  
The overall size and depth of the circuit are $\bo{\frac{nd}{r}+ \frac{d\log d}{k}+\frac{n2^r}{r}}$ and $\bo{\frac{d\log nk}{k}+r}$, respectively. 
Let $r = \ta{\log m/n}$ and $k = \ta{\frac{m}{n}\log m/n}$ such that \cref{m-ineq1} and \cref{m-ineq} are satisfied. Note that from \cref{m-ineq1} and \cref{m-ineq}, $m$ is required to satisfy $m\geq 6n$ since $r \geq 1$ and $k \geq 1$. When $ m \in [6n, \bo{\frac{nd}{\log d}}]$, by plugging the expression of $r$ and $k$ into the overall complexities, the circuit size and depth are $\bo{\frac{nd}{\log m/n}}$ and $\bo{\frac{nd\log m}{m\log m/n}}$, respectively, as claimed. When $m = \omega\rbra{\frac{nd}{\log d}}$, we use only $\ta{\frac{nd}{\log d}}$ ancillary qubits, and then the size and depth are $\bo{\frac{nd}{\log d}}$ and $\bo{\log nd}$. Combining the above two cases, we conclude that, for any $m \geq 6n$, there exists a quantum circuit of size $\bo{\frac{nd}{\log m/n} + \frac{nd}{\log d}}$ and depth $\bo{\frac{nd\log m}{m\log m/n}+\log nd}$ for SQSP utilizing $m$ ancillary qubits. 
    
\end{proof}

\subsection{Deferred Proofs}\label{sec:proof}

We first prove \cref{lemma:encodeunary}.
\begin{proof}[Proof of \cref{lemma:encodeunary}]
The transformation can be implemented through the following steps, where we introduce additional $k$ tag qubits to reduce the circuit depth. By tracing out these tag qubits finally, we obtain the transformation described in \cref{equ:tounary} as follows:  
\begin{align}
    &\ket{j}_{{\ell_n}-\ell_k}\ket{0}_{\ell_k} \ket{e_i}_k \underbrace{\ket{0}_k}_{\text{tag qubits}} \ket{0}_k\\
    \xrightarrow[]{} 
    &\ket{j}_{{\ell_n}-\ell_k}\ket{0}_{\ell_k} \ket{e_i}_k \ket{1}_1\ket{0}_{k-1} \ket{0}_k \label{t1}\\
    \xrightarrow[]{} 
    &\ket{j}_{{\ell_n}-\ell_k}\ket{0}_{\ell_k} \ket{e_i}_k \ket{1\dots1}_k \ket{0}_k \label{t2}\\
    \xrightarrow[]{} 
    &\ket{0}_{{\ell_n}} \ket{e_i}_k \ket{1\dots1}_k \ket{0}_k \label{t3}\\
    \xrightarrow[]{} 
    &\ket{0}_{{\ell_n}} \ket{e_i}_k \ket{1\dots1}_k \ket{e_i}_k \label{t4}\\
    \xrightarrow[]{} 
    &\ket{0}_{{\ell_n}} \ket{0}_k \ket{1}_1\ket{0}_{k-1} \ket{e_i}_k \label{t5}\\
    \xrightarrow[]{} 
    &\ket{0}_{{\ell_n}} \ket{0}_k \ket{0}_k \ket{e_i}_k. \label{t6}
\end{align}
\begin{itemize}
    \item In \cref{t1}, we flip a tag qubit when the first ${\ell_n} - \ell_k$ qubits are in state $\ket{j}$, which, therefore, could be realized by a circuit consisting of a multi-control Toffoli gate and at most ${\ell_n} - \ell_k$ NOT gates.  
    \item \cref{t2} can be achieved via a fanout operation.  
    \item \cref{t3} is to uncompute $\ket{j}_{{\ell_n}-\ell_k}$ and could be implemented by a fanout gate with at most ${\ell_n} - \ell_k$ target qubits. The control qubit is any tag qubit and the target qubits correspond to the positions of $1$ in the binary representation of $j$. 
    \item \cref{t4} is executed using $k$ Toffoli gates. For $1 \leq \ell \leq k$, the control qubits for the $\ell$-th Toffoli gate is the $\ell$-th qubit in $\ket{e_i}_k$ and the $\ell$-th tag qubit, and the target qubit is the $\ell$-th qubit in the last $k$ qubits.  
    \item \cref{t5} firstly uncomputes the first $\ket{e_i}_k$ in \cref{t4} by applying $k$ CNOT gates, each with one of the last $k$ qubits as the control and the corresponding qubit in the first $\ket{e_i}_k$ as the target. Secondly, we uncompute the fanout operation, ensuring that only one tag qubit remains in state $\ket{1}$.  
    \item \cref{t6} is to perform $k$ CNOT gates, each with one of the last $k$ qubits as the control qubit and the remaining activated tag qubit as the target qubit.  
\end{itemize}
The above transformation can be efficiently implemented with an equivalent circuit of size $O({\ell_n} - \ell_k + k)$ and depth $O(\log ({\ell_n} - \ell_k) + \log k)$, according to \cref{lemma:parallelcnot}, \cref{lemma:copy} and \cref{lemma:toffoli}. It could be verified that the constructed circuit acts trivially on all basis states corresponding to $\mathcal{S}_{j'}$ for $j' \neq j$, which completes the proof.
\end{proof}

In the following, we prove \cref{lemma:u2nr}.

\begin{proof}[Proof of \cref{lemma:u2nr}]
    Without loss of generality, we assume that $\abs{\mathcal{X}} = k$.
    The transformation will be implemented through the following two steps with $\frac{nk}{r}+k$ ancillary qubits :
    \begin{align}
        \ket{e_i}_k  \ket{0}_{n2^r/r} 
        & \xrightarrow{\texttt{{Step a}}} \ket{e_i}_k \rbra*{\bigotimes_{j=0}^{\frac{n}{r}-1} \ket{e_{x_i[jr: (j+1)r]}}} \\
        & \xrightarrow{\texttt{{Step b}}} \ket{0}_k \rbra*{\bigotimes_{j=0}^{\frac{n}{r}-1} \ket{e_{x_i[jr: (j+1)r]}}}
    \end{align}
    The first register consisting of $k$ qubits is called Register $\texttt{A}$. The second register consisting of $\frac{n2^r}{r}$ qubits is called Register $\texttt{B}$, which is further split into $\frac{n}{r}$ sub-registers $\cbra{\texttt{B}_{j}}_{0 \leq j \leq \frac{n}{r}-1}$, each of size $2^r$.

\paragraph{\texttt{{Step a}}} Let $\text{CNOT}(i, j, \ell)$ denote a CNOT gate with the $i$-th qubit of Register $\texttt{A}$ as the control qubit and the $\ell$-th qubit of Register $\texttt{B}_{j}$ as the target qubit. An example of $\text{CNOT}(i, j, \ell)$ is shown in \cref{fig:sub1}. \texttt{{Step a}} can be realized with a circuit consisting of CNOT gates as follows:
    \begin{align}
        \prod_{i = 0}^{k-1} \prod_{j = 0}^{\frac{n}{r}-1} \text{CNOT}(i, j, x_i[jr: (j+1)r]) = \prod_{j = 0}^{\frac{n}{r}-1} \prod_{\ell=0}^{2^r-1} \prod_{i \colon x_i[jr: (j+1)r] = \ell} \text{CNOT}(i, j, \ell).
    \end{align}
The above equation holds because all the CNOT gates involved commute with each other, since their control qubits are in Register $\texttt{A}$ and are never used as target qubits.

\begin{itemize}
    \item For fixed $j$ and $\ell$, the product $\prod_{i \colon x_i[jr: (j+1)r] = \ell} \text{CNOT}(i, j, \ell)$ contains at most $k$ CNOT gates. These gates have distinct control qubits and a common target qubit, allowing implementation with an equivalent circuit of depth $\bo{\log k}$, according to \cref{lemma:parallelcnot}.
    
    \item  For fixed $ j $,  the CNOT circuits $\prod_{i \colon x_i[jr: (j+1)r] = \ell} \text{CNOT}(i, j, \ell)$ for different values of $\ell$ have different control qubits and different target qubits, and thus can be implemented in parallel. 
    \item The CNOT circuits $\prod_{\ell=0}^{2^r-1} \prod_{i \colon x_i[jr: (j+1)r] = \ell} \text{CNOT}(i, j, \ell)$ for different values of $j$ share the same control qubits in Register $\texttt{A}$. To execute these operations in parallel, we duplicate $\frac{n}{r}$ copies of the content of Register $\texttt{A}$ in the ancillary qubits using a circuit of size $\bo{\frac{nk}{r}}$ and depth $\bo{\log \frac{n}{r}}$, according to \cref{lemma:copy}. Consequently, all CNOT gates in \texttt{{Step a}} can be implemented by a circuit with depth $\bo{\log k}$ and size $\bo{\frac{nk}{r}}$. Finally, we uncompute the ancillary qubits in depth $\bo{\log \frac{n}{r}}$ and size $\bo{\frac{nk}{r}}$. 
\end{itemize}

\texttt{{Step a}} is summarized as follows:
    \begin{align}
        & \ket{e_i}_{\texttt{A}} \ket{0^{\frac{n2^r}{r}}}_{\texttt{B}} \ket{0^{\frac{nk}{r}}}  \notag \\
        \xrightarrow{} & \ket{e_i}_{\texttt{A}}  \ket{0^{\frac{n2^r}{r}}}_{\texttt{B}} \underbrace{\ket{e_i} \cdots \ket{e_i}}_{\text{$\frac{n}{r}$ copies}} & \text{(\cref{lemma:copy}, depth $\bo{\log \frac{n}{r}}$, size $\bo{\frac{nk}{r}}$)}\\
        \xrightarrow{} & \ket{e_i}_{\texttt{A}}   \rbra*{\bigotimes_{j=0}^{\frac{n}{r}-1} \ket{e_{x_i[jr: (j+1)r]}}_{\texttt{B}_{j}}} \ket{e_i} \cdots \ket{e_i} & \text{(\cref{lemma:parallelcnot}, depth $\bo{\log k}$, size $\bo{\frac{nk}{r}}$)}\\
        \xrightarrow{} & \ket{e_i}_{\texttt{A}}  \rbra*{\bigotimes_{j=0}^{\frac{n}{r}-1} \ket{e_{x_i[jr: (j+1)r]}}_{\texttt{B}_{j}}} \ket{0^{\frac{nk}{r}}} & \text{(\cref{lemma:copy}, depth $\bo{\log \frac{n}{r}}$, size $\bo{\frac{nk}{r}}$)}.
    \end{align}

\paragraph{\texttt{{Step b}}}
    For $0 \leq i \leq d - 1$, let $\text{Tof}_i$ be a $\rbra{\frac{n}{r}+1}$-qubit Toffoli gate acting on Register $\texttt{A}$ and Register $\texttt{B}$, and the control qubits are the $x_i[jr: (j+1)r]$-th qubits of register $\texttt{B}_{j}$ with $0 \leq j \leq \frac{n}{r}-1$ and the target qubit is the $i$-th qubit of Register $\texttt{A}$. An example of $\text{Tof}_i$ is shown in \cref{fig:sub2}. Conditioned on the state of Register $\texttt{B}$ encodes the $(n, r)$-unary representation of $x_i$, $\text{Tof}_i$ flips the $i$-th qubit in Register $\texttt{A}$. Therefore, \texttt{{Step b}} can be implemented by $\prod_{i=0}^{k-1}\text{Tof}_i$. According to \cref{lemma:toffoli}, each $\text{Tof}_i$ can be implemented in depth $\bo{\log \frac{n}{r}}$ and size $\bo{\frac{n}{r}}$ with one ancillary qubit. 
    
    Note that each $\text{Tof}_i$ has a different target qubit, but the sets of their control qubits may intersect. Therefore, we duplicate the content of Register $\texttt{B}$ to parallelize the execution of all the Toffoli gates. Suppose the $\ell$-th qubit in Register $\texttt{B}_{j}$ acts as a control qubit for a total of $t_{j,\ell}$ Toffoli gates. We have $t_{j,\ell} = \sum_i \delta_{x_i[jr: (j+1)r] = \ell}$ and $\sum_\ell t_{j,\ell} = k$. For each $j$ and $\ell$, we create $t_{j,\ell}$ copies of the content of the $\ell$-th qubit of Register $\texttt{B}_{j}$ in the ancillary qubits. This step could be achieved by a circuit of depth $\bo{\log k}$ and size $\bo{\frac{nk}{r}}$ according to \cref{lemma:copy}. Now, we can implement all the Toffoli gates $\cbra{\text{Tof}_i}_{0\leq i \leq k-1}$ in parallel using a circuit of depth $\bo{\log \frac{n}{r}}$ and size $\bo{\frac{nk}{r}}$ with additional $k$ ancillary qubits. Finally, we uncompute the ancillary qubits as before. 
    
    In summary, \texttt{{Step a}} and \texttt{{Step b}} can be implemented by a circuit of depth $\bo{\log k + \log \frac{n}{r}} = \bo{\log nk}$ and size $\bo{\frac{nk}{r}}$, using $\frac{nk}{r}+k$ ancillary qubits. It can be verified that if Register $\texttt{A}$ is in the state $\ket{0^k}$, then \texttt{{Step a}} does not alter the state of any register. Additionally, if the integer encoded in Register $\texttt{B}$ is not in $\mathcal{S}_j$ or Register $\texttt{B}$ is in state $\ket{0}_{n2^r/r}$, then \texttt{{Step b}} also does not alter the state of any register. Therefore, the constructed circuit satisfies \cref{equ:cc22} and \cref{equ:cc23}, completing the proof.
\end{proof}

\begin{figure}[htbp]
  \centering
  \resizebox{0.7\textwidth}{!}{
  \begin{subfigure}[b]{0.45\textwidth}
    \centering
    \includegraphics[width=\linewidth, trim=13cm 6.5cm 13cm 4.5cm, clip]{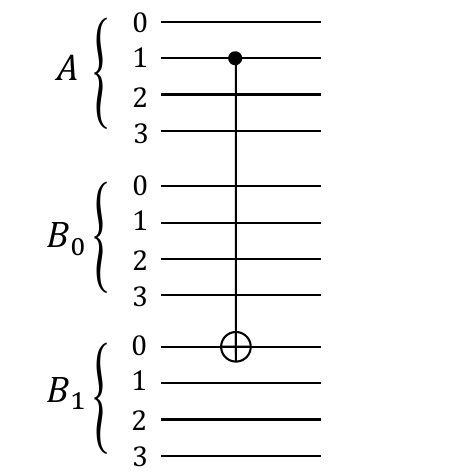}
    \caption{$\text{CNOT}(1, 1, 0)$ with $x_1[2:4] = 00$.}
    \label{fig:sub1}
  \end{subfigure}
  \hfill
  \begin{subfigure}[b]{0.45\textwidth}
    \centering
    \includegraphics[width=\linewidth, trim=13cm 6.5cm 13cm 4.5cm, clip]{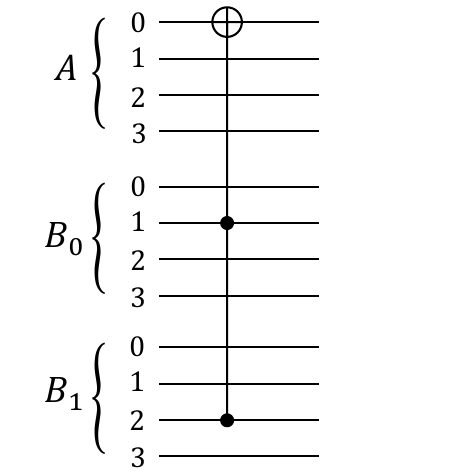}
    \caption{$\text{Tof}_0$ with $x_0 = 0110$.}
    \label{fig:sub2}
  \end{subfigure}
  }
  \caption{Diagrams for $\text{CNOT}(i, j, \ell)$ and $\text{Tof}_i$ with $k = 4, n = 4, r = 2$.}
  \label{fig:combined}
\end{figure}

\section{Conclusion}\label{sec:conclusion}

This paper presents an algorithm for sparse quantum state preparation that optimizes the trade-off between the ancilla count and depth. Given $m$ ancillary qubits, we construct a circuit for SQSP, achieving a circuit size of ${O}(\frac{nd}{\log m/n} + \frac{nd}{\log d})$ and a depth of ${O}(\frac{nd \log m}{m \log m/n} + \log nd)$. When $m = \Theta(\frac{nd}{\log d})$ ancillary qubits are used, our algorithm recovers the optimal circuit depth ${O}(\log nd)$ while reducing the circuit size to ${O}(\frac{nd}{\log d})$. This leads to significant improvements over prior works by reducing both the circuit size and the number of ancillary qubits by a factor of $\log^2 d$.

Our results provide a near-optimal solution for SQSP, with the circuit depth lower bound established as $\Omega\left(\frac{nd}{(m+n)(\log (m+n) + \log d)} + \log nd\right)$. The proposed method demonstrates the potential for optimizing quantum circuits using ancillary qubits, and we believe the techniques can be extended to other quantum circuit synthesis problems, such as the implementation of sparse Boolean functions.

\bibliographystyle{unsrt} 
\bibliography{refs}

\begin{thebibliography}{10}

\bibitem{childs2018toward}
Andrew~M Childs, Dmitri Maslov, Yunseong Nam, Neil~J Ross, and Yuan Su.
\newblock Toward the first quantum simulation with quantum speedup.
\newblock {\em Proceedings of the National Academy of Sciences}, 115(38):9456--9461, 2018.

\bibitem{low2017optimal}
Guang~Hao Low and Isaac~L Chuang.
\newblock Optimal hamiltonian simulation by quantum signal processing.
\newblock {\em Physical Review Letters}, 118(1):010501, 2017.

\bibitem{low2019hamiltonian}
Guang~Hao Low and Isaac~L Chuang.
\newblock Hamiltonian simulation by qubitization.
\newblock {\em Quantum}, 3:163, 2019.

\bibitem{berry2015simulating}
Dominic~W Berry, Andrew~M Childs, Richard Cleve, Robin Kothari, and Rolando~D Somma.
\newblock Simulating hamiltonian dynamics with a truncated taylor series.
\newblock {\em Physical Review Letters}, 114(9):090502, 2015.

\bibitem{lloyd2014quantum}
Seth Lloyd, Masoud Mohseni, and Patrick Rebentrost.
\newblock Quantum principal component analysis.
\newblock {\em Nature Physics}, 10(9):631--633, 2014.

\bibitem{kerenidis2016quantum}
Iordanis Kerenidis and Anupam Prakash.
\newblock Quantum recommendation systems.
\newblock In {\em Proceedings of the 8th Innovations in Theoretical Computer Science Conference}, volume~67 of {\em Leibniz International Proceedings in Informatics (LIPIcs)}, pages 49:1--49:21, 2017.

\bibitem{kerenidis2019q}
Iordanis Kerenidis, Jonas Landman, Alessandro Luongo, and Anupam Prakash.
\newblock q-means: A quantum algorithm for unsupervised machine learning.
\newblock {\em Advances in Neural Information Processing Systems}, 32, 2019.

\bibitem{kerenidis2021quantum}
Iordanis Kerenidis and Jonas Landman.
\newblock Quantum spectral clustering.
\newblock {\em Physical Review A}, 103(4):042415, 2021.

\bibitem{rebentrost2014quantum}
Patrick Rebentrost, Masoud Mohseni, and Seth Lloyd.
\newblock Quantum support vector machine for big data classification.
\newblock {\em Physical Review Letters}, 113(13):130503, 2014.

\bibitem{shende2005synthesis}
Vivek~V Shende, Stephen~S Bullock, and Igor~L Markov.
\newblock Synthesis of quantum logic circuits.
\newblock In {\em Proceedings of the 2005 Asia and South Pacific Design Automation Conference}, pages 272--275, 2005.

\bibitem{sun2023asymptotically}
Xiaoming Sun, Guojing Tian, Shuai Yang, Pei Yuan, and Shengyu Zhang.
\newblock Asymptotically optimal circuit depth for quantum state preparation and general unitary synthesis.
\newblock {\em IEEE Transactions on Computer-Aided Design of Integrated Circuits and Systems}, 42(10):3301--3314, 2023.

\bibitem{zhang2022quantum}
Xiao-Ming Zhang, Tongyang Li, and Xiao Yuan.
\newblock Quantum state preparation with optimal circuit depth: Implementations and applications.
\newblock {\em Physical Review Letters}, 129(23):230504, 2022.

\bibitem{zhang2024circuit}
Xiao-Ming Zhang and Xiao Yuan.
\newblock Circuit complexity of quantum access models for encoding classical data.
\newblock {\em npj Quantum Information}, 10(1):42, 2024.

\bibitem{zhang2024practical}
Xiao-Ming Zhang.
\newblock Practical, optimal preparation of general quantum state with exponentially improved robustness.
\newblock {\em arXiv preprint arXiv:2411.02782}, 2024.

\bibitem{gui2024spacetime}
Kaiwen Gui, Alexander~M Dalzell, Alessandro Achille, Martin Suchara, and Frederic~T Chong.
\newblock Spacetime-efficient low-depth quantum state preparation with applications.
\newblock {\em Quantum}, 8:1257, 2024.

\bibitem{low2024trading}
Guang~Hao Low, Vadym Kliuchnikov, and Luke Schaeffer.
\newblock Trading t gates for dirty qubits in state preparation and unitary synthesis.
\newblock {\em Quantum}, 8:1375, 2024.

\bibitem{gosset2024quantum}
David Gosset, Robin Kothari, and Kewen Wu.
\newblock Quantum state preparation with optimal t-count.
\newblock {\em arXiv preprint arXiv:2411.04790}, 2024.

\bibitem{holmes2020efficient}
Adam Holmes and Anne~Y Matsuura.
\newblock Efficient quantum circuits for accurate state preparation of smooth, differentiable functions.
\newblock In {\em Proceedings of IEEE International Conference on Quantum Computing and Engineering}, pages 169--179. IEEE, 2020.

\bibitem{melnikov2023quantum}
Ar~A Melnikov, Alena~A Termanova, Sergey~V Dolgov, Florian Neukart, and MR~Perelshtein.
\newblock Quantum state preparation using tensor networks.
\newblock {\em Quantum Science and Technology}, 8(3):035027, 2023.

\bibitem{iaconis2024quantum}
Jason Iaconis, Sonika Johri, and Elton~Yechao Zhu.
\newblock Quantum state preparation of normal distributions using matrix product states.
\newblock {\em npj Quantum Information}, 10(1):15, 2024.

\bibitem{gonzalez2024efficient}
Javier Gonzalez-Conde, Thomas~W Watts, Pablo Rodriguez-Grasa, and Mikel Sanz.
\newblock Efficient quantum amplitude encoding of polynomial functions.
\newblock {\em Quantum}, 8:1297, 2024.

\bibitem{manabe2025state}
Hidetaka Manabe and Yuichi Sano.
\newblock The state preparation of multivariate normal distributions using tree tensor network.
\newblock {\em Quantum}, 9:1755, 2025.

\bibitem{gundlapalli2022deterministic}
Prithvi Gundlapalli and Junyi Lee.
\newblock Deterministic and entanglement-efficient preparation of amplitude-encoded quantum registers.
\newblock {\em Physical Review Applied}, 18(2):024013, 2022.

\bibitem{mcardle2022quantum}
Sam McArdle, Andr{\'a}s Gily{\'e}n, and Mario Berta.
\newblock Quantum state preparation without coherent arithmetic.
\newblock {\em arXiv preprint arXiv:2210.14892}, 2022.

\bibitem{rosenkranz2025quantum}
Matthias Rosenkranz, Eric Brunner, Gabriel Marin-Sanchez, Nathan Fitzpatrick, Silas Dilkes, Yao Tang, Yuta Kikuchi, and Marcello Benedetti.
\newblock Quantum state preparation for multivariate functions.
\newblock {\em Quantum}, 9:1703, 2025.

\bibitem{zylberman2024efficient}
Julien Zylberman and Fabrice Debbasch.
\newblock Efficient quantum state preparation with walsh series.
\newblock {\em Physical Review A}, 109(4):042401, 2024.

\bibitem{moosa2023linear}
Mudassir Moosa, Thomas~W Watts, Yiyou Chen, Abhijat Sarma, and Peter~L McMahon.
\newblock Linear-depth quantum circuits for loading fourier approximations of arbitrary functions.
\newblock {\em Quantum Science and Technology}, 9(1):015002, 2023.

\bibitem{marin2023quantum}
Gabriel Marin-Sanchez, Javier Gonzalez-Conde, and Mikel Sanz.
\newblock Quantum algorithms for approximate function loading.
\newblock {\em Physical Review Research}, 5(3):033114, 2023.

\bibitem{nakaji2022approximate}
Kouhei Nakaji, Shumpei Uno, Yohichi Suzuki, Rudy Raymond, Tamiya Onodera, Tomoki Tanaka, Hiroyuki Tezuka, Naoki Mitsuda, and Naoki Yamamoto.
\newblock Approximate amplitude encoding in shallow parameterized quantum circuits and its application to financial market indicators.
\newblock {\em Physical Review Research}, 4(2):023136, 2022.

\bibitem{zoufal2019quantum}
Christa Zoufal, Aur{\'e}lien Lucchi, and Stefan Woerner.
\newblock Quantum generative adversarial networks for learning and loading random distributions.
\newblock {\em npj Quantum Information}, 5(1):103, 2019.

\bibitem{sanders2019black}
Yuval~R Sanders, Guang~Hao Low, Artur Scherer, and Dominic~W Berry.
\newblock Black-box quantum state preparation without arithmetic.
\newblock {\em Physical Review Letters}, 122(2):020502, 2019.

\bibitem{bausch2022fast}
Johannes Bausch.
\newblock Fast black-box quantum state preparation.
\newblock {\em Quantum}, 6:773, 2022.

\bibitem{wang2021fast}
Shengbin Wang, Zhimin Wang, Guolong Cui, Shangshang Shi, Ruimin Shang, Lixin Fan, Wendong Li, Zhiqiang Wei, and Yongjian Gu.
\newblock Fast black-box quantum state preparation based on linear combination of unitaries.
\newblock {\em Quantum Information Processing}, 20(8):270, 2021.

\bibitem{rattew2022preparing}
Arthur~G Rattew and B{\'a}lint Koczor.
\newblock Preparing arbitrary continuous functions in quantum registers with logarithmic complexity.
\newblock {\em arXiv preprint arXiv:2205.00519}, 2022.

\bibitem{bergholm2005quantum}
Ville Bergholm, Juha~J Vartiainen, Mikko M{\"o}tt{\"o}nen, and Martti~M Salomaa.
\newblock Quantum circuits with uniformly controlled one-qubit gates.
\newblock {\em Physical Review A}, 71(5):052330, 2005.

\bibitem{grover2002creating}
Lov Grover and Terry Rudolph.
\newblock Creating superpositions that correspond to efficiently integrable probability distributions.
\newblock {\em arXiv preprint quant-ph/0208112}, 2002.

\bibitem{plesch2011quantum}
Martin Plesch and {\v{C}}aslav Brukner.
\newblock Quantum-state preparation with universal gate decompositions.
\newblock {\em Physical Review A}, 83(3):032302, 2011.

\bibitem{iten2016quantum}
Raban Iten, Roger Colbeck, Ivan Kukuljan, Jonathan Home, and Matthias Christandl.
\newblock Quantum circuits for isometries.
\newblock {\em Physical Review A}, 93(3):032318, 2016.

\bibitem{yuan2023optimal}
Pei Yuan and Shengyu Zhang.
\newblock Optimal (controlled) quantum state preparation and improved unitary synthesis by quantum circuits with any number of ancillary qubits.
\newblock {\em Quantum}, 7:956, 2023.

\bibitem{rosenthal2021query}
Gregory Rosenthal.
\newblock Query and depth upper bounds for quantum unitaries via grover search.
\newblock {\em arXiv preprint arXiv:2111.07992}, 2021.

\bibitem{clader2022quantum}
B~David Clader, Alexander~M Dalzell, Nikitas Stamatopoulos, Grant Salton, Mario Berta, and William~J Zeng.
\newblock Quantum resources required to block-encode a matrix of classical data.
\newblock {\em IEEE Transactions on Quantum Engineering}, 3:1--23, 2022.

\bibitem{harrow2009quantum}
Aram~W Harrow, Avinatan Hassidim, and Seth Lloyd.
\newblock Quantum algorithm for linear systems of equations.
\newblock {\em Physical Review Letters}, 103(15):150502, 2009.

\bibitem{costa2022optimal}
Pedro~CS Costa, Dong An, Yuval~R Sanders, Yuan Su, Ryan Babbush, and Dominic~W Berry.
\newblock Optimal scaling quantum linear-systems solver via discrete adiabatic theorem.
\newblock {\em PRX quantum}, 3(4):040303, 2022.

\bibitem{martyn2021grand}
John~M Martyn, Zane~M Rossi, Andrew~K Tan, and Isaac~L Chuang.
\newblock Grand unification of quantum algorithms.
\newblock {\em PRX quantum}, 2(4):040203, 2021.

\bibitem{gilyen2019quantum}
Andr{\'a}s Gily{\'e}n, Yuan Su, Guang~Hao Low, and Nathan Wiebe.
\newblock Quantum singular value transformation and beyond: exponential improvements for quantum matrix arithmetics.
\newblock In {\em Proceedings of the 51st annual ACM SIGACT symposium on theory of computing}, pages 193--204, 2019.

\bibitem{gleinig2021efficient}
Niels Gleinig and Torsten Hoefler.
\newblock An efficient algorithm for sparse quantum state preparation.
\newblock In {\em Proceedings of the 58th ACM/IEEE Design Automation Conference}, pages 433--438. IEEE, 2021.

\bibitem{malvetti2021quantum}
Emanuel Malvetti, Raban Iten, and Roger Colbeck.
\newblock Quantum circuits for sparse isometries.
\newblock {\em Quantum}, 5:412, 2021.

\bibitem{ramacciotti2023simple}
Debora Ramacciotti, Andreea~I Lefterovici, and Antonio~F Rotundo.
\newblock Simple quantum algorithm to efficiently prepare sparse states.
\newblock {\em Physical Review A}, 110(3):032609, 2024.

\bibitem{mozafari2022efficient}
Fereshte Mozafari, Giovanni De~Micheli, and Yuxiang Yang.
\newblock Efficient deterministic preparation of quantum states using decision diagrams.
\newblock {\em Physical Review A}, 106(2):022617, 2022.

\bibitem{de2020circuit}
Tiago~ML de~Veras, Ismael~CS De~Araujo, Daniel~K Park, and Adenilton~J da~Silva.
\newblock Circuit-based quantum random access memory for classical data with continuous amplitudes.
\newblock {\em IEEE Transactions on Computers}, 70(12):2125--2135, 2020.

\bibitem{de2022double}
Tiago~ML de~Veras, Leon~D da~Silva, and Adenilton~J da~Silva.
\newblock Double sparse quantum state preparation.
\newblock {\em Quantum Information Processing}, 21(6):204, 2022.

\bibitem{mao2024towards}
Rui Mao, Guojing Tian, and Xiaoming Sun.
\newblock Toward optimal circuit size for sparse quantum state preparation.
\newblock {\em Physical Review A}, 110(3):032439, 2024.

\bibitem{luo2024circuit}
Jingquan Luo and Lvzhou Li.
\newblock Circuit complexity of sparse quantum state preparation.
\newblock {\em arXiv preprint arXiv:2406.16142}, 2024.

\bibitem{yeo2025reducing}
Hyeonjun Yeo, Ha~Eum Kim, IlKwon Sohn, and Kabgyun Jeong.
\newblock Reducing circuit depth in quantum state preparation for quantum simulation using measurements and feedforward.
\newblock {\em Physical Review Applied}, 23(5):054066, 2025.

\bibitem{buhrman2024state}
Harry Buhrman, Marten Folkertsma, Bruno Loff, and Niels~MP Neumann.
\newblock State preparation by shallow circuits using feed forward.
\newblock {\em Quantum}, 8:1552, 2024.

\bibitem{zi2025constant}
Wei Zi, Junhong Nie, and Xiaoming Sun.
\newblock Constant-depth quantum circuits for arbitrary quantum state preparation via measurement and feedback.
\newblock {\em arXiv preprint arXiv:2503.16208}, 2025.

\bibitem{sun2021asymptotically}
Xiaoming Sun, Guojing Tian, Shuai Yang, Pei Yuan, and Shengyu Zhang.
\newblock Asymptotically optimal circuit depth for quantum state preparation and general unitary synthesis.
\newblock {\em arXiv preprint arXiv:2108.06150}, 2021.

\bibitem{moore2001parallel}
Cristopher Moore and Martin Nilsson.
\newblock Parallel quantum computation and quantum codes.
\newblock {\em SIAM Journal on Computing}, 31(3):799--815, 2001.

\bibitem{nie2024quantum}
Junhong Nie, Wei Zi, and Xiaoming Sun.
\newblock Quantum circuit for multi-qubit toffoli gate with optimal resource.
\newblock {\em arXiv preprint arXiv:2402.05053}, 2024.

\end{thebibliography}


\end{document}